\documentclass[a4paper,11pt]{amsart}
\usepackage{amsmath,amssymb}
\usepackage{mathrsfs}
\usepackage{eucal}
\usepackage{cases}
\usepackage{cite}
\usepackage{comment}

\makeatletter
 
  \@addtoreset{equation}{section}
 \makeatother

\newcommand{\slim}{\operatorname{s-}\hspace{-0.25cm}\lim_{t\to\pm\infty}}

\newcommand{\Schr}{Schr\"odinger }

\renewcommand{\a}{\alpha}

\renewcommand{\d}{\delta}

\renewcommand{\th}{\theta}

\renewcommand{\k}{\kappa}
\renewcommand{\l}{\lambda}

\newcommand{\x}{\xi}

\newcommand{\C}{\mathbb{C}}
\newcommand{\R}{\mathbb{R}}
\newcommand{\Z}{\mathbb{Z}}
\newcommand{\T}{\mathbb{T}}

\newcommand{\la}{\langle}
\newcommand{\ra}{\rangle}

\newtheorem{thm}{Theorem}[section]
\newtheorem{lem}[thm]{Lemma}
\newtheorem{prop}[thm]{Proposition}

\theoremstyle{definition}

\newtheorem{ass}[thm]{Assumption}

\theoremstyle{remark}
\newtheorem{rem}[thm]{Remark}

\title[Long-range scattering on graphene]{Long-range scattering theory for discrete Schr\"odinger operators on graphene}
\author{Yukihide TADANO}
\thanks{
Graduate School of Mathematical Sciences, the University of Tokyo, 3-8-1 Komaba, Meguro-ku, Tokyo, 153-8914, Japan. \\
E-mail: {\tt tadano@ms.u-tokyo.ac.jp}}
\date{December 21, 2018}

\begin{document}
\begin{abstract}
We consider a long-range scattering theory for discrete Schr\"odinger operators on the hexagonal lattice, which describe tight-binding Hamiltonians on the graphene sheet.
We construct Isozaki-Kitada modifiers for a pair of the difference Laplacian on the hexagonal lattice and perturbed operators with long-range potentials. We prove that these modified wave operators exist and that they are asymptotically complete.
\end{abstract}

\maketitle


\section{Introduction}  
The aim of the present paper is to consider a long-range scattering theory for discrete Schr\"odinger operators on graphene, that is, the hexagonal lattice.
Unlike discrete Schr\"odinger operators on the square and triangular lattices, operators on the hexagonal lattice cannot be represented as an operator on the space of $\C^1$-valued functions on $\mathbb{Z}^2$, but $\C^2$-valued.
Because of this aspect, a long-range scattering theory for this model cannot be treated as in \cite{T}. 
In the present paper, we generalize the results of \cite{T}, and in particular we construct Isozaki-Kitada modifiers for the hexagonal lattice.
For a short-range scattering theory for discrete Schr\"odinger operators on general lattices, including the hexagonal lattice, see \cite{P-R}.
See also \cite{An-I-M} and \cite{An-I-M2}  for spectral properties of discrete \Schr operators on general lattices.

Let $\mathcal{H}=\ell^2(\mathbb{Z}^2; \mathbb{C}^2)$. For $u \in \mathcal{H}$, we use the notation $u=\left( \begin{array}{c}
u_1 \\ u_2
\end{array}
\right)$, $u_1,u_2 \in \ell^2(\mathbb{Z}^2)$. The unperturbed discrete Schr\"odinger operator $H_0$ on graphene is described as the negative of the difference Laplacian
\begin{align}
H_0u[x] = -\left( \begin{array}{c}
u_2[x]+u_2[x-e_1]+u_2[x-e_2] \\ u_1[x]+u_1[x+e_1]+u_1[x+e_2]
\end{array}
\right), \quad x \in \mathbb{Z}^2, u \in \mathcal{H},
\end{align}
where $e_1=(1,0),$ $e_2=(0,1)$. The derivation of $H_0$ is found in e.g.\ \cite{A} and \cite{An-I-M}.
As is seen later, $H_0$ has purely absolutely continuous spectrum and $\sigma(H_0) = \sigma_{ac}(H_0) = [-3,3]$.

For a function $V:\mathbb{Z}^2 \to \mathbb{R}^2$, the corresponding multiplication operator is also denoted by $V$:
\begin{align}
Vu[x] = \left( \begin{array}{c}
V_1[x] u_1[x] \\ V_2[x] u_2[x]
\end{array}
\right) , 
\end{align}
where $V_1[x]$ and $V_2[x]$ are the first and second value of $V[x]$.
We set $H=H_0+V$.
It is known that if $V$ is short-range, i.e., $|V[x]|\leq C \la x \ra^{-\rho}$ for some $C>0$ and $\rho>1$, where $\la x\ra := (1+|x|^2)^{\frac{1}{2}}$, the wave operators
\begin{align*}
W^\pm = \slim e^{itH} e^{-itH_0}
\end{align*}
exist and $W^\pm$ are asymptotically complete, i.e., the range of $W^\pm$ equals to the absolutely continuous subspace of $H$.
In the present paper, we assume the long-range condition below.

\begin{ass}\label{ass}
The function $V$ has the following representation
\begin{align*}
V_1=V_\ell+V_{s,1}, \ V_2=V_\ell+V_{s,2},
\end{align*}
where $V_\ell$ and $V_{s,j}$ satisfy
\begin{align*}
&| \tilde\partial^\alpha V_\ell [x] | \leq C_\alpha \langle x \rangle^{-|\alpha|-\rho} , \quad x\in\mathbb{Z}^2, \ \alpha\in\mathbb{Z}_+^2 , \\
&| V_{s,j}[x] | \leq C \langle x \rangle^{-1-\rho} , \quad x\in\mathbb{Z}^2, \ j=1,2
\end{align*}
for some $\rho\in (0,1]$ and $C_\alpha , C > 0$. Here $\tilde\partial^\alpha=\tilde\partial_{x_1}^{\alpha_1}\tilde\partial_{x_2}^{\alpha_2}$, $\tilde\partial_{x_j} W[x]:=W[x]-W[x-e_j]$.
\end{ass}

\begin{rem}
The above assumption is invariant under the choice of isomorphism between the set of vertices of the hexagonal lattice and $\Z^2 \times \{1,2\}$ invariant under the canonical $\Z^2$ action.
In particular, it follows that the difference between each pair of the nearest vertices is short-range.
We note that the pair of potentials $V_1[x]=\langle x \rangle^{-1}$ and $V_2[x]=-\langle x \rangle^{-1}$, an analog of Wigner von-Neumann potentials, is not allowed under the above assumption.
We also note that, for $1$-dimensional discrete \Schr operators, embedded eigenvalues can occur even if $|V[x]| \leq C\la x\ra^{-1}$, $x\in\Z$ for some $C>0$ (see \cite{L}).
\end{rem}


We give notations for the description of the main theorem.
For a selfadjoint operator $S$ and an Borel set $I \subset \R$, $E_S(I)$ denotes the spectral projection of $S$ onto $I$ and $\mathcal{H}_{ac}(S)$ denotes the absolutely continuous subspace of $S$.
The main theorem of this paper is the following.

\begin{thm}\label{Main}
Assume that $V$ satisfies Assumption \ref{ass}. Then for any open set $\Gamma \Subset [-3,3] \backslash \{0,\pm 1, \pm 3\}$, one can construct a Fredholm operator $J$ on $\mathcal{H}$ such that there exist modified wave operators
\begin{align} \label{mWO}
W_{J}^\pm(\Gamma) := \slim e^{itH} J e^{-itH_0} E_{H_0}(\Gamma)
\end{align}
and the following properties hold:

i)Intertwining property $H W_{J}^\pm(\Gamma)=W_{J}^\pm(\Gamma) H_0$.

ii)Partial isometries $\|W_{J}^\pm(\Gamma) u\|=\|E_{H_0}(\Gamma)u\|$.

iii)Asymptotic completeness $\operatorname{Ran} W_{J}^\pm(\Gamma)=E_{H}(\Gamma)\mathcal{H}_{\text{ac}}(H)$.

\end{thm}

The above theorem is an analog of \cite{N} and \cite{T} in the sense of a long-range scattering theory for discrete Schr\"odinger operators.
For a long-range scattering theory for Schr\"odinger operators on the Euclidean space, see e.g.\ \cite{D-G}, \cite{Y} and references therein.

 We observe spectral properties of the free operator $H_0$, and we show an abstract form of the operator $J$ in (\ref{mWO}).
By $\mathcal{F}:\mathcal{H} \rightarrow L^2(\mathbb{T}^2; \mathbb{C}^2)$, $\mathbb{T}^2:=[-\pi,\pi)^2$, we denote the discrete Fourier transform
\begin{align}
&\mathcal{F}u(\xi)=\left(\begin{array}{c}
Fu_1(\xi) \\
Fu_2(\xi)
\end{array}
\right), \quad \xi \in \mathbb{T}^2, \\
&Fu_j(\xi)=(2\pi)^{-1} \sum_{x\in\mathbb{Z}^2} e^{-ix\cdot\xi} u_j[x], \quad j=1,2. \nonumber
\end{align}
Then $\mathcal{F} \circ H_0 \circ \mathcal{F}^*$ is a multiplier by the matrix
\begin{align}
H_0(\xi) = \left( \begin{array}{cc}
0 & \overline{\alpha(\xi)} \\
\alpha(\xi) & 0
\end{array}
\right),
\end{align}
where $\alpha(\xi):=-(1+e^{i\xi_1}+e^{i\xi_2})$.
Note that for each $\xi \in \mathbb{T}^2$, $H_0(\xi)$ is an Hermitian matrix.

In order to determine the spectrum $\sigma(H_0)$ of $H_0$, we consider the diagonalization of matrix at each point in the momentum space $\mathbb{T}^2$.
We set a unitary matrix
\begin{align*}
U(\xi):= \frac{1}{\sqrt{2}} \begin{pmatrix}
1 & -\frac{\overline{\alpha(\xi)}}{|\alpha(\xi)|} \\ \frac{\alpha(\xi)}{|\alpha(\xi)|} & 1
\end{pmatrix}, \quad \xi \in \mathbb{T}^2\backslash \{ \alpha^{-1}(0) \}.
\end{align*}
Then $H_0(\xi)$ is diagonalized by $U(\xi)$; setting $p(\xi):=|\alpha(\xi)|$, we have
\begin{align*}
\tilde H_0(\xi):= U(\xi)^* H_0(\xi) U(\xi) = \begin{pmatrix}
p(\xi) & 0 \\ 0 & -p(\xi)
\end{pmatrix}, \ \xi \in \mathbb{T}^2\backslash \{ \alpha^{-1}(0) \}.
\end{align*}
Since $\alpha^{-1}(0)= \{(\pm \frac{2}{3}\pi, \mp \frac{2}{3}\pi)\}$, $\tilde H_0(\xi)$ and $U(\xi)$ are defined a.e.\ in $\mathbb{T}^2$.
Furthermore $p$ is smooth outside $\alpha^{-1} (0)$ and the set of its critical points
\begin{align}
\operatorname{Cr} :=& \{ \xi \in \mathbb{T}^2 \backslash \alpha^{-1}(0) \mid \nabla_\xi p(\xi)=0\} \\
=&\{(0,0),(-\pi,0),(0,-\pi),(-\pi,-\pi)\} \nonumber
\end{align}
has Lebesgue measure zero.
Thus $H_0$ has purely absolutely continuous spectrum (see \cite[Proposition 1.2]{T}) and
\begin{align}
\sigma(H_0) = \overline{p(\mathbb{T}^2\backslash\alpha^{-1}(0)) \cup \left( - p(\mathbb{T}^2\backslash\alpha^{-1}(0)) \right)} = [-3,3].
\end{align}

Using the above $U$, $J$ is formally represented as
\begin{align*}
J = \mathcal{F}^* U(\cdot) \mathcal{F} \circ
\begin{pmatrix} J_+ & 0 \\ 0 & J_- \end{pmatrix}
\circ \mathcal{F}^* U(\cdot)^* \mathcal{F} ,
\end{align*}
where
\begin{align*}
J_\pm u[x] = (2\pi)^{-1} \int_{\T^2} e^{i\varphi_{\pm}(x,\x)} Fu(\x) d\x
\end{align*}
and $\varphi_\pm(x,\x)$, $(x,\x)\in \R^2 \times \T^2$, are outgoing and incoming solution of the eikonal equation
\begin{align*}
p(\nabla_x\varphi) + \tilde V_\ell(x) = p(\x) ,
\end{align*}
where $\tilde V_\ell$ is a suitable smooth extension of $V_\ell$ onto $\R^2$.
However there are two technical difficulties. One is the singularity of $p(\x)$ at $\alpha^{-1}(0)$. The other is the singularity of $U(\x)$.
The latter is more crucial because we cannot prove that the difference $V_\ell - \mathcal{F}^* U \mathcal{F} \circ V_\ell \circ \mathcal{F}^* U^* \mathcal{F}$ is short-range due to the singularity of $U(\x)$.
We will avoid the above difficulties in Subsection 2.1.

We describe the outline of this paper.
 The essential idea of proof is as follows; in order to make the above long-range scattering problem easier, we replace the free operator $H_0$ to a modified free operator $H_0^\prime$ witch can be diagonalized in the whole momentum space $\mathbb{T}^2$.
 In Subsection \ref{subsec2.1}, we construct the modified free operator $H^\prime_0$, and the property of $H^\prime_0$ is written in Lemma \ref{step1}.
 Considering the long-range scattering theory for $H_0^\prime$ instead of $ H_0$, we can reduce the problem of long-range scattering for operators on $\mathcal{H}$ to that for operators on $\ell^2(\mathbb{Z}^2)$. Then we will see in Subsection \ref{long-range} that the result of \cite{T} is applicable to the above setting.
 Subsection \ref{short-range} concerns a short-range scattering theory.
 This step is treated with the limiting absorption principle and Kato's smooth perturbation theory.
 We also use a pseudodifferential calculus prepared in Subsection \ref{PD subsection}.
  In Appendix \ref{Mourre section}, we show the limiting absorption principle by using the Mourre theory.

\section*{Acknowledgement}
The author was supported by JSPS Research Fellowship for Young Scientists, KAKENHI Grant Number 17J05051. The author would like to thank his supervisor Shu Nakamura for encouraging to write this paper and helpful discussions about the construction of modified wave operators.

\section{Preliminaries} \label{Sketch section}

\subsection{Construction of the modified free operator}\label{subsec2.1}

Let us fix the open interval $\Gamma \Subset [-3,3] \backslash \{0,\pm 1, \pm 3\}$ as in Theorem \ref{Main} and $\delta:= \operatorname{dist}(0,\Gamma) = \inf_{\lambda\in\Gamma} |\lambda|$.
 We construct a modified free operator 
\begin{align*}
H^\prime_0:=\mathcal{F}^* \circ H^\prime_0(\cdot) \circ \mathcal{F},
\end{align*}
where $H^\prime_0(\xi) \in C^\infty(\mathbb{T}^2; M_2(\mathbb{C}))$ is a simple symmetric matrix for each $\xi\in\mathbb{T}^2$. We will choose $H^\prime_0(\x)$ so that $H^\prime_0$ has the same spectral projection as $H_0$ in $[-3,-\frac{\d}{2}]\cup [\frac{\d}{2},3]$.


Let $\kappa \in C^\infty\left( \mathbb{R}_{\geq 0} ; \mathbb{R}_{\geq 0} \right)$ be fixed such that $\operatorname{supp} \kappa \subset [ 0 , \frac{\delta^2}{4} )$ and $0 < E+\kappa(E)^2 < \frac{\delta^2}{4}$ for $E \in [0,\frac{\delta^2}{4})$.
Let us define
%
%
%
\begin{align}
H^\prime(\xi) := \begin{pmatrix} \kappa( p(\xi)^2 ) & \overline{\alpha(\xi)} \\ \alpha(\xi) & -\kappa( p(\xi)^2 ) \end{pmatrix}
\end{align}
Then $H^\prime(\xi)$ has two eigenvalues
\begin{align} \label{ev}
\lambda_\pm (\xi) := \pm \left( \kappa(p(\xi)^2)^2 + p(\xi)^2 \right)^{\frac{1}{2}}
\end{align}
and the corresponding eigenvectors are
\begin{align*}
f_+(\xi) = 
\begin{pmatrix} \kappa(p(\xi)^2) + \lambda_+ (\xi) \\ \alpha(\xi)
\end{pmatrix}, 
\ 
f_-(\xi) = 
\begin{pmatrix} -\overline{\alpha(\xi)} \\ \kappa(p(\xi)^2) + \lambda_+ (\xi) 
\end{pmatrix} .
\end{align*}
Therefore letting
\begin{align}
U^\prime(\xi) :=& 
\frac{1}{C(\xi)} 
\begin{pmatrix}
\kappa(p(\xi)^2) + \lambda_+ (\xi) & -\overline{\alpha(\xi)} \\ 
\alpha(\xi) & \kappa(p(\xi)^2) + \lambda_+ (\xi)
\end{pmatrix}, \label{U prime} \\
C(\xi) :=& \left\{ p(\xi)^2 + \left[ \kappa(p(\xi)^2) + \lambda_+ (\xi) \right]^2 \right\}^{\frac{1}{2}}, \nonumber
\end{align}
we learn that $U^\prime (\xi)$ is a unitary matrix-valued smooth function on $\mathbb{T}^2$ and
\begin{align}
\tilde H^\prime_0 (\xi) := U^\prime (\xi)^* H^\prime_0 (\xi) U^\prime (\xi)
= 
\begin{pmatrix}
	\lambda_+(\xi) & 0 \\
	0 & \lambda_- (\xi) 
\end{pmatrix}
.
\end{align}
Note that $\l_\pm(\x)=\pm p(\x)$ for $\x\in p^{-1}\left((\frac{\delta}{2},3]\right)$ by the condition of $\k$. 
Thus we obtain the following lemma.

\begin{lem}\label{step1}
Let
\begin{align}\label{def of H prime}
H^\prime_0 := \mathcal{F}^* H^\prime_0(\cdot) \mathcal{F}, \
 \tilde H^\prime_0 := \mathcal{F}^* \tilde H^\prime_0(\cdot) \mathcal{F}, \
 U^\prime := \mathcal{F}^* U^\prime(\cdot) \mathcal{F}
\end{align}
and
\begin{align}\label{def of lambda}
\lambda_\pm := F^* \lambda_\pm(\cdot) F .
\end{align}
Then
\begin{align}
&\tilde H^\prime_0 = (U^\prime)^* H^\prime_0 U^\prime
= \begin{pmatrix} \lambda_+ & 0 \\ 0 & \lambda_- \end{pmatrix} , \\
&E_{H_0}(I) = E_{H^\prime_0}(I) , \quad I \Subset (-\infty, -\frac{\delta}{2}) \cup (\frac{\delta}{2}, \infty).
\end{align}
In particular,
\begin{align}
& e^{-itH_0} E_{H_0}\left(\Gamma\right) = e^{-itH^\prime_0} E_{H^\prime_0}\left(\Gamma\right) , \quad t\in\mathbb{R} , \\
& \chi(H_0) = \chi(H^\prime_0) , \quad \chi \in C^\infty_c(\Gamma) .
\end{align}
\end{lem}

\subsection{Pseudodifferential calculus}\label{PD subsection}
In this subsection we prepare an assertion concerning the boundedness of pseudodifference operators on $\Z^d$, $d\geq1$.
This lemma is an analog of symbol calculus of pseudodifferential operators on $\mathbb{T}^2$.
The proof is given in Appendix \ref{pfPD}.
See also \cite[Theorem 4.7.10]{R-T}.

\begin{lem} \label{PD calculus}
Let $m_1,m_2 \in\R$, $a : \T^d \times \Z^d \to \C$, $b : \Z^d \to \C$, and
\begin{align*}
\mathrm{Op}(a)u[x] &= (2\pi)^{-d} \int_{\T^d} \sum_{y\in\Z^d} e^{i(x-y)\cdot\x} a(\x,y) u[y] d\x , \\
\mathrm{Op}(b)u[x] &= b[x]u[x].
\end{align*}
Suppose that $a(\cdot,y) \in C^\infty(\T^d)$ for $y \in \Z^d$ and $|\partial_\xi^\alpha a(\x,y)| \leq C_\alpha \la y\ra^{-m_1}$, $|\tilde \partial_{x_j} b[x]| \leq C \langle x \rangle^{-m_2}$ for $x\in\Z^2$, $j=1,\dots, d$, where $\tilde \partial_{x_j} b[x] = b[x] - b[x-e_j]$.
Then, $\la x \ra^{p} [\mathrm{Op}(b), \mathrm{Op}(a)] \la x \ra^{q}$ is a bounded operator on $\ell^2(\Z^d)$ if $p+q=m_1+m_2$.
\end{lem}

\section{Proof of Theorem \ref{Main}}

First note that the properties i) and ii) are satisfied if the limits (\ref{mWO}) exist. See \cite{IsoKita} and \cite{Y} for the proofs.

We denote by $V_\ell$ the multiplication operator by $\begin{pmatrix} V_\ell[x] \\ V_\ell[x] \end{pmatrix}$ if there is no risk of confusion.
Let
\begin{align}
H^\prime_\ell := H^\prime_0 + U^\prime V_\ell \left(U^\prime\right)^* = U^\prime \left( \tilde H_0^\prime + V_\ell \right) \left(U^\prime\right)^*.
\end{align}
Then it suffices to prove the following two assertions.

\begin{thm} \label{Step2}
One can construct a Fredholm operator $J$ such that there exist modified wave operators
\begin{align}\label{mWO in step2}
W_{J,\ell}^{\prime,\pm}(\Gamma) := \slim e^{itH^\prime_\ell} J e^{-itH^\prime_0} E_{H^\prime_0}(\Gamma)
\end{align}
exist and $\operatorname{Ran} W_{J,\ell}^{\prime,\pm}(\Gamma) = E_{H^\prime_\ell}(\Gamma) \mathcal{H}_{ac}(H^\prime_\ell)$.
\end{thm}

\begin{thm}\label{Step3}
There exist the wave operators
\begin{align}\label{WO}
W_s^{\prime,\pm}(\Gamma) := \slim e^{itH} e^{-itH^\prime_\ell} E_{H^\prime_\ell}^{\text{ac}}(\Gamma),
\end{align}
where $E_{H^\prime_\ell}^{\text{ac}}(\Gamma)$ denotes the spectral projection of $H^\prime_\ell$ onto the absolutely continuous subspace in $\Gamma$, and $\operatorname{Ran} W_s^{\prime,\pm}(\Gamma) = E_{H}(\Gamma) \mathcal{H}_{\text{ac}}(H)$.
\end{thm}

\begin{proof}[Proof of Theorem \ref{Main}]
It remains to prove $W_{J}^{\pm}(\Gamma) = W_s^{\prime,\pm}(\Gamma) W_{J,\ell}^{\prime,\pm}(\Gamma)$.
For $u\in \mathcal{H}$, it follows from Lemma \ref{step1} that
\begin{align*}
e^{itH} J e^{-itH_0} E_{H_0}(\Gamma) u
 =& e^{itH} J e^{-itH_0^\prime} E_{H_0^\prime}(\Gamma) u \\
 =& e^{itH} e^{-itH^\prime_\ell} \cdot e^{itH^\prime_\ell} J e^{-itH_0^\prime} E_{H_0^\prime}(\Gamma) u .
\end{align*}
Note that by Theorem \ref{Step2} there exist $T_\pm > 0$ such that for $\pm t > T_\pm$, $e^{itH^\prime_\ell} J e^{-itH_0^\prime} E_{H_0^\prime}(\Gamma) u = W_{J,\ell}^{\prime,\pm}(\Gamma) u + r_\pm(t)$ and $\| r_\pm(t) \|_\mathcal{H} \to 0$ as $t\to\pm\infty$.
Thus we have for $\pm t > T_\pm$
\begin{align}\label{err}
& \| e^{itH} J e^{-itH_0} E_{H_0}(\Gamma) u - W_s^{\prime,\pm}(\Gamma) W_{J,\ell}^{\prime,\pm}(\Gamma) u \|_\mathcal{H} \\
 \leq& \left\| \left(e^{itH} e^{-itH^\prime_\ell} - W_s^{\prime,\pm}(\Gamma)\right) W_{J,\ell}^{\prime,\pm}(\Gamma)u \right\|_\mathcal{H} + \| r_\pm(t)) \|_\mathcal{H} . \nonumber
\end{align}
Since $W_{J,\ell}^{\prime,\pm}(\Gamma)u \in \operatorname{Ran} W_{J,\ell}^{\prime,\pm}(\Gamma) = E_{H^\prime_\ell}(\Gamma) \mathcal{H}_{ac}(H^\prime_\ell)$, Theorem \ref{Step3} implies that (\ref{err}) tends to $0$ as $t\to\pm\infty$.
\end{proof}

In the following, we prove Theorems \ref{Step2} and \ref{Step3}.

\subsection{Proof of Theorem \ref{Step2}}\label{long-range}
 We reduce a long-range scattering problem on $\mathcal{H}=\ell^2(\mathbb{Z}^2;\mathbb{C}^2)$ into that on $\ell^2(\Z^2)$, which  is considered in \cite{T}.

The existence and completeness of (\ref{mWO in step2}) are equivalent to those of
\begin{align}
\tilde W_{\tilde J, \ell}^{\prime,\pm}(\Gamma)&:=
\slim e^{it\tilde H_\ell^\prime} \tilde J e^{-it\tilde H_0^\prime} E_{\tilde H_0^\prime}(\Gamma) ,
\end{align}
where $\tilde J = \left(U^\prime\right)^* J U^\prime$ and
\begin{align}
\tilde H_\ell^\prime &:= \left(U^\prime\right)^* H_\ell^\prime U^\prime = \tilde H_0^\prime + V_\ell .
\end{align}
Indeed, a direct calculus implies
\begin{align} \label{WO to dWO}
W_{J,\ell}^{\prime,\pm}(\Gamma)= U^\prime \tilde W_{\tilde J, \ell}^{\prime,\pm}(\Gamma) \left(U^\prime\right)^*.
\end{align}
Set $\tilde J = \begin{pmatrix} \tilde J_+ & 0 \\ 0 & \tilde J_- \end{pmatrix}$, $\tilde J_\pm \in \mathcal{B}(\ell^2(\mathbb{Z}^2))$.
 Then the scattering problem of operators on $\mathcal{H}$ is reduced to that on $\ell^2(\mathbb{Z}^2)$:
\begin{align}
& \tilde W_{\tilde J, \ell}^\pm(\Gamma) \\
=& \slim 
\begin{pmatrix}
	e^{it(\lambda_+ + V_\ell)} \tilde J_+ e^{-it\lambda_+} E_{\lambda_+}(\Gamma) & 0 \\
	0 & e^{it(\lambda_- + V_\ell)} \tilde J_- e^{-it\lambda_-} E_{\lambda_-}(\Gamma)
\end{pmatrix}
. \nonumber
\end{align}
Therefore we obtain the following theorem by \cite{T}.

\begin{thm} \label{diagonalized WO}
There exist Fredholm operators $\tilde J_\#$, $\# \in \{ +,-\}$, of the form
\begin{align}\label{tildeJ}
\tilde J_\# v[x] = (2\pi)^{-1} \int_{\mathbb{T}^2} e^{i\varphi_\#(x,\xi)} Fv(\xi) d\xi, \quad v \in \ell^2(\mathbb{Z}^2),
\end{align}
such that the modified wave operators
\begin{align}
W^\pm_{\ell, \#}(\Gamma) := \slim e^{it(\lambda_\# + V_\ell)} \tilde J_\# e^{-it\lambda_\#} E_{\lambda_\#}(\Gamma)
\end{align}
exist and they are partial isometries from $E_{\lambda_\#} (\Gamma) \mathcal{H}_{ac}(\lambda_\#)$ onto $E_{\lambda_\#} (\Gamma) \mathcal{H}_{ac}(\lambda_\#)$.
\end{thm}


Note that each $\varphi_\#(x,\xi)$ in (\ref{tildeJ}) is constructed as a smooth function on $\mathbb{R}^2 \times \mathbb{T}^2$ which solves the eikonal equation
\begin{align}
\lambda_\#(\nabla_x \varphi_\#(x,\xi)) + \tilde V_\ell (x) = \lambda_\#(\xi)
\end{align}
on the outgoing and incoming regions, where $\tilde V_\ell \in C^\infty(\R^2)$ is a suitable extension of $V_\ell$. For detailed properties of $J_\pm$ and $\varphi_\pm$, see \cite{T}.

Let $J:= U^\prime \tilde J \left( U^\prime \right)^*$, $\tilde J := \begin{pmatrix} \tilde J_+ & 0 \\ 0 & \tilde J_- \end{pmatrix}$.
Then using Theorem \ref{diagonalized WO} and (\ref{WO to dWO}), we obtain Theorem \ref{Step2}.
\qed

\subsection{Proof of Theorem \ref{Step3}}\label{short-range}

Theorem \ref{Step3} is proved by Proposition \ref{Kato smooth} and the Cook-Kuroda method. The proof of the next proposition is given in Appendix \ref{Mourre section}.

\begin{prop} \label{Kato smooth}
$i)$ $H$ and $H^\prime_\ell$ have at most finite discrete eigenvalues in $\Gamma$ with counting their multiplicities. \\
$ii)$ Let $s > \frac{1}{2}$ and $\chi \in C_{c}^\infty(\Gamma \backslash \sigma_{\text{pp}}(H))$ $($resp.\ $\chi \in C_{c}^\infty(\Gamma \backslash \sigma_{\text{pp}}(H^\prime_\ell)) )$. Then $\langle x\rangle^{-s} \chi(H)$ $($resp.\ $\langle x\rangle^{-s} \chi(H^\prime_\ell))$ is $H$$($resp.\ $H^\prime_\ell)$-smooth.
\end{prop}

According to Proposition \ref{Kato smooth} i) and a density argument, it suffices to show the existence of wave operators
\begin{align}
&\lim_{t\to\pm\infty} e^{itH} e^{-itH^\prime_\ell} u , \label{direct WO} \\
&\lim_{t\to\pm\infty} e^{itH^\prime_\ell} e^{-itH} v \label{inverse WO}
\end{align}
for $u\in \mathcal{H}_{\text{ac}}(H^\prime_\ell)$ and $v\in \mathcal{H}_{\text{ac}}(H)$ such that
\begin{align}\label{cutoff}
\chi(H^\prime_\ell)u = u, \quad \psi(H)v = v
\end{align}
with $\chi \in C_{c}^\infty(\Gamma \backslash \sigma_{\text{pp}}(H^\prime_\ell))$ and $\psi \in C_{c}^\infty(\Gamma \backslash \sigma_{\text{pp}}(H))$.
We prove the existence of (\ref{direct WO}) as $t\to\infty$ only. The other cases are proved similarly.

By (\ref{cutoff}), we have
\begin{align} \label{compactness argument}
e^{itH} e^{-itH^\prime_\ell} u =& e^{itH} \chi(H^\prime_\ell)^3 e^{-itH^\prime_\ell} u \\
=& e^{itH} \chi(H)\chi(H^\prime_0)\chi(H^\prime_\ell) e^{-itH^\prime_\ell} u \nonumber \\
&+ e^{itH} \left( \chi(H^\prime_\ell)^2-\chi(H)\chi(H^\prime_0) \right)\chi(H^\prime_\ell) e^{-itH^\prime_\ell} u . \nonumber 
\end{align}
Since $H-H_0=V_\ell$ and $H^\prime_\ell - H^\prime_0=\left(U^\prime\right)^* V_\ell U^\prime$ are compact operators, the Helffer-Sj\"ostrand formula implies that $\chi(H) - \chi(H_0) = \chi(H) - \chi(H^\prime_0)$ and $\chi(H^\prime_\ell)-\chi(H^\prime_0)$ are compact. Thus $\chi(H^\prime_\ell)^2-\chi(H)\chi(H^\prime_0)$ is also a compact operator.
Note that $u\in\mathcal{H}_{ac}(H^\prime_\ell)$ implies $e^{-itH^\prime_\ell}u \to 0$ weakly as $t\to\infty$.
Thus the last term of (\ref{compactness argument}) converges to $0$ as $t\to\infty$, and 
it suffices to prove the existence of the limit
\begin{align*}
\lim_{t\to\infty} e^{itH} \chi(H)\chi(H^\prime_0)\chi(H^\prime_\ell) e^{-itH^\prime_\ell} u.
\end{align*}

Now we use the Cook-Kuroda method. First note that
\begin{align*}
&e^{itH} \chi(H)\chi(H^\prime_0)\chi(H^\prime_\ell) e^{-itH^\prime_\ell} u
 - e^{it^\prime H} \chi(H)\chi(H^\prime_0)\chi(H^\prime_\ell) e^{-it^\prime H^\prime_\ell} u \\
&= \int_{t^\prime}^t \frac{d}{ds}(e^{isH} \chi(H)\chi(H^\prime_0)\chi(H^\prime_\ell) e^{-isH^\prime_\ell} u) ds .
\end{align*}
A direct calculus implies
\begin{align*}
&\frac{d}{dt}\left(e^{itH} \chi(H)\chi(H^\prime_0)\chi(H^\prime_\ell) e^{-itH^\prime_\ell} u \right) \\
&= i e^{itH} \chi(H) \left(H \chi(H^\prime_0) - \chi(H^\prime_0) H^\prime_\ell \right) \chi(H^\prime_\ell) e^{-itH^\prime_\ell} u \\
&= i e^{itH} \chi(H) \left(V_s\chi(H^\prime_0) + V_\ell \chi(H^\prime_0) - \chi(H^\prime_0) \left(U^\prime\right)^* V_\ell U^\prime \right) \chi(H^\prime_\ell) e^{-itH^\prime_\ell} u \\
&= i e^{itH} \chi(H) \left( V_s\chi(H^\prime_0) + [ V_\ell , \chi(H^\prime_0) \left(U^\prime\right)^*] U^\prime \right) \chi(H^\prime_\ell) e^{-itH^\prime_\ell} u \\
&= i e^{itH} (A_1^* B_1 + A_2^* B_2) e^{-itH^\prime_\ell} u ,
\end{align*}
where $\gamma:=\frac{1+\rho}{2}$ and
\begin{align*}
A_1&:=\langle x\rangle^{\gamma} V_s\chi(H), \quad
B_1:=\langle x\rangle^{-\gamma} \chi(H^\prime_0) \chi(H^\prime_\ell), \\
A_2&:=\langle x\rangle^{-\gamma} \chi(H), \quad
B_2:=\langle x\rangle^{\gamma} [ V_\ell , \chi(H^\prime_0) \left(U^\prime\right)^*] U^\prime \chi(H^\prime_\ell) .
\end{align*}
Then by a standard argument of short-range scattering theory (see, e.g., \cite{R-S}), it suffices to prove that each $A_j (B_j)$ is $H (H^\prime_\ell)$-smooth, respectively.
The $H$-smoothness of $A_1$ and $A_2$ is a direct consequence of Proposition \ref{Kato smooth}.
For the $H^\prime_\ell$-smoothness of $B_1$ and $B_2$, note that
\begin{align*}
B_1&=\langle x\rangle^{-\gamma} \chi(H^\prime_0) \langle x\rangle^{\gamma}
\cdot \langle x\rangle^{-\gamma} \chi(H^\prime_\ell) , \\
B_2&=\langle x\rangle^{\gamma} [ V_\ell , \chi(H^\prime_0) \left(U^\prime\right)^*] U^\prime \langle x\rangle^{\gamma} \cdot \langle x\rangle^{-\gamma} \chi(H^\prime_\ell) .
\end{align*}
Then it follows from Lemma \ref{PD calculus} and Assumption \ref{ass} that $\langle x \rangle^{-\gamma} \chi(H^\prime_0) \langle x \rangle^{\gamma}$ and $\langle x \rangle^{\gamma} [ V_\ell , \chi(H^\prime_0) \left(U^\prime\right)^*] U^\prime \langle x \rangle^{\gamma}$ are bounded.
Combining this and Proposition \ref{Kato smooth}, we obtain the $H^\prime_\ell$-smoothness of $B_1$ and $B_2$.
\qed

\appendix

\section{Mourre theory for $H$ and $H^\prime_\ell$, and the proof of Proposition \ref{Kato smooth}} \label{Mourre section}

In this appendix, we review the Mourre theory, the limiting absorption principle and Kato's smooth perturbation theory.
Let $\Gamma \Subset \sigma(H_0)\backslash \{0,\pm1,\pm3\}$ be an open interval as in Theorem \ref{Main}.
For a selfadjoint operator $A$ and $n \in \mathbb{N}$, let
\begin{align*}
C^n(A) = \{ S \in \mathcal{B}(\mathcal{H}) \mid \mathbb{R} \to \mathcal{B}(\mathcal{H}), t \mapsto e^{-itA} S e^{itA} \text{ is strongly of class $C^n$} \} ,
\end{align*}
and $C^\infty(A) := \cap_{n\in\mathbb{N}} C^n(A)$.
We denote by $\mathcal{C}^{1,1}(A)$ the set of the operators satisfying
\begin{align}
\int_0^1 \| e^{-itA} S e^{itA} + e^{itA} S e^{-itA} - 2S \| \frac{dt}{t^2} <\infty.
\end{align}
Note that $C^2(A) \subset \mathcal{C}^{1,1}(A) \subset C^1(A)$.
We denote by $B$ the Besov space $(\mathcal{D}(A),\mathcal{H})_{\frac{1}{2},1}$ obtained by real interpolation. We also denote by $B^*$ its dual. 
The definition of real interpolation is found in \cite[Section 2.3]{A-BdM-G}.

We recall the characterization of Kato smoothness; for a selfadjoint operator $H$ and an $H$-bounded operator $G$, we say that $G$ is $H$-smooth if
\begin{align}\label{K smooth1}
C_1 := \frac{1}{2\pi} \sup_{u\in D(H),\|u\|=1} \int_{\mathbb{R}} \left\|G e^{-itH} u \right\| dt < \infty .
\end{align}
It is known that there are other characterizations of $H$-smoothness and one of them is
\begin{align}\label{K smooth 2}
C_2 := \sup_{\lambda\in\R,\varepsilon>0} \| G \delta(\lambda,\varepsilon) G^* \| < \infty ,
\end{align}
moreover $C_1=C_2$, where $\delta(\lambda,\varepsilon):=\frac{1}{2\pi i} \left( (H-\lambda-i\varepsilon)^{-1} - (H-\lambda+i\varepsilon)^{-1} \right)$.
For other characterizations, see \cite{Kato}.

In order to prove Proposition \ref{Kato smooth}, we apply the two operators $H$ and $H^\prime_\ell$ to Theorem \ref{Mourre thm} described below with $I \Subset \Gamma$ 
 and a suitable conjugate operator $A$.
The following theorem is a standard result of the Mourre theory and is due to \cite[Proposition 7.1.3, Corollary 7.2.11, Theorem 7.3.1]{A-BdM-G}.

\begin{thm} \label{Mourre thm}
Let $S \in \mathcal{C}^{1,1}(A)$ and $I \subset \mathbb{R}$ be an open interval.
Suppose that there exist a constant $c>0$ and a compact operator $K$ on $\mathcal{H}$ such that
\begin{align} \label{Mourre ineq}
E_{S}(I) [S,iA] E_{S}(I) \geq c E_{S}(I) + K .
\end{align}
Then \\
i) $S$ has at most a finite number of eigenvalues in $I$ and each eigenvalues in $I$ has finite multiplicity. \\
ii) For any $\lambda \in I \backslash \sigma_{pp}(S)$, there exist the weak-* limits in $\mathcal{B}(B,B^*)$
\begin{align}
\operatorname{w*-}\lim_{\varepsilon \to +0} (S-\lambda \mp i\varepsilon)^{-1} ,
\end{align}
and the convergence is locally uniform in $I \backslash \sigma_{pp}(S)$.
In particular, for any $I^\prime \Subset I \backslash \sigma_{pp}(S)$, $S$ is purely absolutely continuous in $I^\prime$ and
\begin{align}
\sup_{\lambda\in I^\prime,\varepsilon>0} \| (S-\lambda \mp i\varepsilon)^{-1} \|_{\mathcal{B}(B,B^*)} <\infty.
\end{align}

\end{thm}

We define the conjugate operator $A$ by
\begin{align}
&A := U^\prime \circ \tilde A \circ \left(U^\prime\right)^* \label{conjugate} \\
&\tilde A := 
\begin{pmatrix} 
	\tilde A_+ & 0 \\ 
	0 & \tilde A_-
\end{pmatrix}
, \label{conjugate 2} \\
&\tilde A_\pm := \frac{1}{2} \sum_{j=1}^2 \left( F^* (\partial_{\xi_j} \lambda_\pm) F \cdot x_j + x_j \cdot F^* (\partial_{\xi_j} \lambda_\pm) F \right), \label{conjugate 3}
\end{align}
where $U^\prime$ and $\lambda_\pm(\x)\in C^\infty(\T^2)$ are given by (\ref{ev}), (\ref{U prime}), (\ref{def of H prime}), (\ref{def of lambda}).
Then Nelson's commutator theorem with the positive selfadjoint operator $\langle x \rangle$ implies that $A$ is essentially selfadjoint on the Schwarz space on $\Z^2$ defined by $\mathcal{S}(\mathbb{Z}^2)= \{ u : \mathbb{Z}^2 \to \mathbb{C}^2 \mid \sup_{x\in\mathbb{Z}^2} \langle x \rangle^{N} |u[x]| < \infty \text{ for any } N \in \mathbb{N} \}$.

First we verify a relation between $A$ and the unperturbed operators $H_0$ and $H^\prime_0$.

\begin{lem}
Both $H_0$ and $H^\prime_0$ belong to $C^\infty(A)$. 
Moreover, let 
\begin{align*}
c:=\min\left\{\inf_{\xi\in \lambda_+^{-1}(\Gamma)} |\nabla_\xi \lambda_+(\xi)|^2, \inf_{\xi\in \lambda_-^{-1}(\Gamma)} |\nabla_\xi \lambda_-(\xi)|^2 \right\}.
\end{align*}
Then, $c>0$ and
\begin{align}
&E_{H_0}(\Gamma) [H_0,iA] E_{H_0}(\Gamma)
\geq c E_{H_0}(\Gamma) , \label{Mourre 1} \\
&E_{H^\prime_0}(\Gamma) [H^\prime_0,iA] E_{H^\prime_0}(\Gamma)
\geq c E_{H^\prime_0}(\Gamma) . \label{Mourre 2}
\end{align}
\end{lem}

\begin{proof}
Note that the LHS (resp.\ RHS) of (\ref{Mourre 1}) equals to the LHS (resp.\ RHS) of (\ref{Mourre 2}) by the construction of $H^\prime_0$.

Since $\mathcal{F}^* H_0 \mathcal{F}$ and $\mathcal{F}^* H^\prime_0 \mathcal{F}$ are multipliers with smooth symbols on $\mathbb{T}^2$ and $\mathcal{F}^* A \mathcal{F}$ is a differential operator of degree $1$ on $\mathbb{T}^2$,
$\mathcal{F}^* [H_0,iA] \mathcal{F}$ and $\mathcal{F}^* [H^\prime_0,iA] \mathcal{F}$ are also multipliers with smooth symbols.
Inductively we see that $H_0, H^\prime_0\in C^\infty(A)$.

For the proof of (\ref{Mourre 2}), a direct calculus implies
\begin{align*}
(U^\prime)^* [H^\prime_0,iA] U &= 
[\tilde H^\prime_0, i\tilde A] \\
&=
\begin{pmatrix}
	|\nabla_\xi \lambda_+ (D_x)|^2 & 0 \\
	0 & |\nabla_\xi \lambda_- (D_x)|^2
\end{pmatrix}
 , \nonumber \\
(U^\prime)^* E_{H^\prime_0}(\Gamma) U^\prime &= E_{\tilde H^\prime_0}(\Gamma) \\
&= 
\begin{pmatrix}
	\chi_{\lambda_+^{-1}(\Gamma)} (D_x) & 0 \\
	0 & \chi_{\lambda_-^{-1}(\Gamma)} (D_x) 
\end{pmatrix}
, \nonumber
\end{align*}
where $\chi_{\lambda_\pm^{-1}(\Gamma)}$ denote the characteristic function of $\lambda_\pm^{-1}(\Gamma)$.
Therefore we obtain
\begin{align*}
&E_{H^\prime_0}(\Gamma) [H^\prime_0,iA] E_{H^\prime_0}(\Gamma) \\
=& U^\prime E_{\tilde H^\prime_0}(\Gamma) [\tilde H^\prime_0,i\tilde A] E_{\tilde H^\prime_0}(\Gamma) (U^\prime)^* \\
=&  U^\prime
\begin{pmatrix}
	|\nabla_\xi \lambda_+ (D_x)|^2 \chi_{\lambda_+^{-1}(\Gamma)} (D_x) & 0 \\
	0 & |\nabla_\xi \lambda_- (D_x)|^2 \chi_{\lambda_-^{-1}(\Gamma)} (D_x)
\end{pmatrix}
(U^\prime)^* \\
\geq &  U^\prime c E_{\tilde H^\prime_0}(\Gamma) (U^\prime)^* \\
= & c E_{H^\prime_0}(\Gamma) .
\end{align*}
\end{proof}

We consider commutators of the perturbations $V_s$, $V_\ell$ and $U^\prime V_\ell \left(U^\prime\right)^*$, and the conjugate operator $A$.
The next lemma claims that the commutators are small in the sense of the Mourre theory, i.e., compact.

\begin{lem}\label{comm}
Let $V_s$ and $V_\ell$ satisfy the condition in Assumption \ref{ass} with $\rho >0$.
Then, for $W=V_s$, $V_\ell$ and $U^\prime V_\ell \left(U^\prime\right)^*$,
\begin{align}\label{C estimate}
\la x \ra^{\rho} [W,iA] \in \mathcal{B}(\mathcal{H}).
\end{align}

\end{lem}

\begin{proof}
First note that
\begin{align} \label{commutator}
\left(U^\prime\right)^*[U^\prime V_\ell \left(U^\prime\right)^* , iA ] U^\prime
= [V_\ell , i \tilde A ] 
= 
\begin{pmatrix}
	[V_\ell , i \tilde A_+ ] & 0 \\ 
	0 & [V_\ell , i \tilde A_- ]
\end{pmatrix}
. 
\end{align}
Since $\tilde A_\pm = \mathrm{Op}(\tilde a_\pm)$ with some functions $\tilde a_\pm$ on $\T^2\times\Z^2$ satisfying the condition of Lemma \ref{PD calculus} with $m_1=1$,
it follows that $\la x \ra^{1+\rho} [V_\ell , i \tilde A_\pm ]$ are bounded on $\ell^2(\mathbb{Z}^2)$. 
Since
\begin{align*}
[U^\prime V_\ell \left(U^\prime\right)^* , iA ]
= U^\prime [V_\ell , i\tilde A] \left( U^\prime \right)^*,
\end{align*}
using Lemma \ref{PD calculus} again shows (\ref{C estimate}) for $W=U^\prime V_\ell \left(U^\prime\right)^*$.
In order to show (\ref{C estimate}) for $W=V_\ell$ or $V_s$, note that $A$ has the representation
\begin{align*}
A = 
\begin{pmatrix}
	\mathrm{Op}(a_{11}) & \mathrm{Op}(a_{12}) \\
	\mathrm{Op}(a_{21}) & \mathrm{Op}(a_{22}) 
\end{pmatrix}
, 
\end{align*}
where each $a_{ij}$ satisfies the condition of Lemma \ref{PD calculus} with $m_1=1$.
Then (\ref{C estimate}) for $W=V_\ell$ is a direct result of Lemma \ref{PD calculus}. 
The last case is also proved since $\la x \ra^\rho V_s A$ and $\la x \ra^\rho A V_s$ are bounded operators by Lemma \ref{PD calculus}.

\end{proof}

Using the above lemma, we see that the perturbations are of the class of $\mathcal{C}^{1,1}(A)$.

\begin{lem}
Let $V_s$ and $V_\ell$ satisfy the condition in Assumption \ref{ass} with $\rho >0$. Then $V_s$, $V_\ell$ and $U^\prime V_\ell \left(U^\prime\right)^*$ 
belong to $\mathcal{C}^{1,1}(A)$.
\end{lem}

\begin{proof}

The proof is motivated by \cite[Lemma 6.2]{P-R}.
First we remark that the operator $\Lambda u[x] := \la x \ra u[x]$ satisfies the condition of \cite[Theorem 6.1]{BdM-S}; 
the first condition is attained by the unitarity of $e^{-it\Lambda}$, $t\in\R$, on $\mathcal{H}$, and the second one that $A^N \Lambda^{-N}$ is a bounded operator on $\mathcal{H}$ for some integer $N \geq 1$, follows from Lemma \ref{PD calculus}.
Thus it suffices to show
\begin{align} \label{C 1,1}
\int_1^\infty \frac{d\l}{\l} \left\| \th\left(\frac{\Lambda}{\l}\right) [W,iA]\right\|_{\mathcal{B}(\mathcal{H})} < \infty
\end{align}
for $W=V_s$, $V_\ell$ and $U^\prime V_\ell \left(U^\prime\right)^*$ and some $\th \in C_c^\infty\left( (0,\infty) \right)$ not identically zero. 
However this follows from Lemma \ref{comm} and
\begin{align*}
\left\| \th\left(\frac{\Lambda}{\l}\right) [W,iA]\right\|_{\mathcal{B}(\mathcal{H})}
&\leq \left\| \th\left(\frac{\Lambda}{\l}\right) \Lambda^{-\rho} \right\|_{\mathcal{B}(\mathcal{H})}
\left\| \Lambda^\rho [W,iA]\right\|_{\mathcal{B}(\mathcal{H})} \\
& \leq C \la \l \ra^{-\rho} \left\| \Lambda^\rho [W,iA]\right\|_{\mathcal{B}(\mathcal{H})}.
\end{align*}

\end{proof}

We have confirmed that Theorem \ref{Mourre thm} is applicable to $S=H$ or $H^\prime_\ell$ and $A$ defined by (\ref{conjugate}), (\ref{conjugate 2}) and (\ref{conjugate 3}).
%
%
Therefore we obtain the limiting absorption principle for $H^\prime_\ell$ and $H$.

\begin{thm} \label{LAP}
i) $H$ and $H^\prime_\ell$ have finitely many eigenvalues with counting multiplicity in $\Gamma$.
\\
ii) For any $I \Subset \Gamma \backslash \sigma_{\text{pp}}(H)$, $I^\prime \Subset \Gamma \backslash \sigma_{\text{pp}}(H^\prime_\ell)$ and $s>\frac{1}{2}$,
\begin{align*}
&\sup_{\lambda\in I,\ \varepsilon>0} \left\| \langle x \rangle^{-s} (H-\lambda\mp i\varepsilon)^{-1} \langle x \rangle^{-s} \right\|_{\mathcal{B}(\mathcal{H})} < \infty, \\
&\sup_{\lambda^\prime\in I^\prime,\ \varepsilon>0} \left\| \langle x \rangle^{-s} (H^\prime_\ell-\lambda^\prime \mp i\varepsilon)^{-1} \langle x \rangle^{-s} \right\|_{\mathcal{B}(\mathcal{H})} < \infty.
\end{align*}
\end{thm}

\begin{proof}

It remains to prove that $\mathcal{H}_s:= \la x \ra^{-s}\mathcal{H} \subset B$ if $s >\frac{1}{2}$.
However it is shown if we remark that $\mathcal{H}_1 \subset \mathcal{D}(A)$ and hence 
$\mathcal{H}_s \subset (\mathcal{H}_1, \mathcal{H})_{\frac{1}{2},1} \subset (\mathcal{D}(A), \mathcal{H})_{\frac{1}{2},1} = B$.
\end{proof}

\begin{proof}[Proof of Proposition \ref{Kato smooth}]
It suffices to show (\ref{K smooth 2}) for $G=\langle x \rangle^{-s} \chi(H)$ and $\langle x \rangle^{-s} \chi(H^\prime_\ell)$.
However this is proved by Theorem \ref{LAP} and the condition of $\operatorname{supp}\chi$.
\end{proof}

\begin{rem}
Theorem \ref{LAP}  may look like a direct consequence of \cite{P-R}.
However the above assertion is more concrete in that the set $\mathcal{T}=\{0,\pm1,\pm3\}$ of threshold energies is explicitly determined.
\end{rem}

\begin{rem}
For any $\lambda \in \Gamma \backslash \sigma_{\text{pp}}(H)$, $\lambda^\prime \in \Gamma \backslash \sigma_{\text{pp}}(\hat{H}_\ell)$ and $s>\frac{1}{2}$, there exist the norm limits
\begin{align*}
&\lim_{\varepsilon\downarrow0} \langle x \rangle^{-s} (H-\lambda\mp i\varepsilon)^{-1} \langle x \rangle^{-s}, \\
&\lim_{\varepsilon\downarrow0} \langle x \rangle^{-s} (\hat{H}_\ell-\lambda^\prime \mp i\varepsilon)^{-1} \langle x \rangle^{-s}.
\end{align*}
For the proof, see e.g.\ \cite{An-I-M} and \cite{N2}.
\end{rem}

\section{Proof of Lemma \ref{PD calculus}}\label{pfPD}
First we observe that
\begin{align*}
\mathrm{Op}(a)u[x] = (2\pi)^{-d} \sum_{y\in\Z^d} A[x,y] u[y], \quad u \in \mathcal{S}(\Z^d) ,
\end{align*} 
where
\begin{align*}
A[x,y] := \int_{\T^d} e^{i(x-y)\cdot\x} a(\x,y) d\x .
\end{align*}
A direct calculus implies
\begin{align*}
(\la \cdot \ra^{p} [\mathrm{Op}(b),\mathrm{Op}(a)] \la \cdot \ra^{q}) u[x] = (2\pi)^{-d} \sum_{y\in\Z^d} K[x,y] u[y],
\end{align*}
where
\begin{align*}
K[x,y] := \la x \ra^{p} \la y \ra^{q} (W[x] - W[y]) A[x,y].
\end{align*}
According to Young's inequality, the boundedness of the operator follows from the inequalities 
\begin{align}
\sup_{y\in\Z^d}\sum_{x\in\Z^d} K[x,y] < \infty ,  \label{Young 1}\\
\sup_{x\in\Z^d}\sum_{y\in\Z^d} K[x,y] < \infty. \label{Young 2}
\end{align}
Let $z:=x-y$, $k:=\sum_{j=1}^d |z_j|$, and $z^0=0,z^1,\dots,z^{k}=z$ be a path from $0$ to $z$ in $\mathbb{Z}^d$ such that $|z^i - z^{i-1}|=1$ for $i=1,2,\dots,k$.
 Then we learn
\begin{align*}
|W[x]-W[y]| &\leq \sum_{i=1}^{k} |W[z^i+y] - W[z^{i-1}+y]| \\
&\leq C \sum_{i=1}^{k} \la z^{i}+y \ra^{-m_2} \\
&= C \sum_{i=1}^{k} \la z^{i}+y \ra^{-p} \la z^{i}+y \ra^{-q+m_1} \\
&\leq C^\prime \sum_{i=1}^{k} \la x-y-z^i \ra^{|p|} \la x \ra^{-p} 
\la z^i \ra^{|q-m_1|} \la y \ra^{-q+m_1} \\
&\leq C^{\prime\prime} \la x-y \ra^{M} \la x \ra^{-p} \la y \ra^{-q+m_1} ,
\end{align*}
where $M:=|p|+|q-m_2|$.
Note that the second last inequality follows from
\begin{align*}
\la x + y \ra &\leq C_d \la x \ra \la y \ra , \\
\la x + y \ra^{-1} &\leq C_d \la x \ra \la y \ra^{-1}
\end{align*}
for $x,y\in\R^d$.
In order to estimate $A[x,y]$, we observe for $\a\in\Z_+^d$
\begin{align*}
|(x-y)^\alpha A[x,y]| = \left| i^{|\alpha|} \int_{\T^d} e^{i(x-y)\cdot\x} \partial_\xi^\alpha a(\x,y) d\x \right| \leq C_{\alpha} \la y \ra^{-m_1} .
\end{align*}
Thus we have
\begin{align*}
|K[x,y]| \leq C^{\prime\prime} \la x-y \ra^{M} \la y\ra^{m_1} |A[x,y]| \leq C_{M,d} \la x-y \ra^{-d-1}.
\end{align*}
Hence we obtain (\ref{Young 1}) and (\ref{Young 2}).
\qed

\end{document}